\documentclass[12pt]{article}
\usepackage{amsmath}
\usepackage{graphicx}
\usepackage{psfrag,epsf}
\usepackage{enumerate}
\usepackage{natbib}
\usepackage{url} % not crucial - just used below for the URL 
\usepackage{comment}

\usepackage[usenames, dvipsnames]{color}
\usepackage{caption}
\usepackage{subcaption}
\usepackage{algorithm}
\usepackage{algorithmicx}
\usepackage{algpseudocode}
\usepackage{amsthm}

\newcommand{\blind}{0}

\newtheorem{thm}{Theorem}[section]
\newtheorem{lem}[thm]{Lemma}

\renewcommand{\v}[1]{\mathbf{#1}}
\newcommand{\vs}[1]{\boldsymbol #1}

\begin{document}

\def\spacingset#1{\renewcommand{\baselinestretch}%
{#1}\small\normalsize} \spacingset{1}

%%%%%%%%%%%%%%%%%%%%%%%%%%%%%%%%%%%%%%%%%%%%%%%%%%%%%%%%%%%%%%%%%%%%%%%%%%%%%%

\if0\blind
{
  \title{\bf Rank conditional coverage and confidence intervals in high dimensional problems}
  \author{Jean Morrison\thanks{
    J.M. and N.S. were supported by NIH Grant DP5OD019820.}\hspace{.2cm}\\
    and \\
    Noah Simon \\
    Department of Biostatistics, University of Washington, Seattle, WA}
  \maketitle
} \fi

\if1\blind
{
  \bigskip
  \bigskip
  \bigskip
  \begin{center}
    {\LARGE\bf Rank conditional coverage and confidence intervals in high dimensional problems}
\end{center}
  \medskip
} \fi

\bigskip
\begin{abstract}
Confidence interval procedures used in low dimensional settings are often inappropriate for high dimensional applications.
When a large number of parameters are estimated,
marginal confidence intervals associated with the most significant estimates have very low coverage rates: They are too small and centered at biased estimates.
The problem of forming confidence intervals in high dimensional settings has previously been studied through the lens of selection adjustment. In this framework, the goal is to control the proportion of non-covering intervals formed for selected parameters.

In this paper we approach the problem by considering the relationship between rank and coverage probability.
Marginal confidence intervals have very low coverage rates for significant parameters and high rates for parameters with more boring estimates. Many selection adjusted intervals display the same pattern.
This connection motivates us to propose a new coverage criterion for confidence intervals in multiple testing/covering problems --- the rank conditional coverage (RCC). This is the expected coverage rate of an interval given the significance ranking for the associated estimator. 
We propose interval construction via bootstrapping which produces small intervals and have a rank conditional coverage close to the nominal level. These methods are implemented in the {\tt R} package {\tt rcc}.
\end{abstract}

\noindent%
\vfill

\newpage
\spacingset{1.45} % DON'T change the spacing!
\section{Introduction}
\label{sec:intro}
In many fields including genomics, proteomics, biomedical science, and neurology it is now common to conduct ``high dimensional'' studies in which thousands or millions of parameters are estimated.
Often, one of the main goals of these studies is to select a small subset of features for description and future investigation. 
Estimates of selected parameters are often reported unadjusted and, when confidence intervals are not entirely omitted, either marginal or Bonferroni-corrected intervals are given. Many previous authors have demonstrated the undesirable features of these practices (\cite{Efron2011}, \cite{Sun2005}, and \cite{Simon2013} among others). In particular, features selected for reporting are usually chosen because they have the largest or most significant estimates; unfortunately these most extreme estimates are also highly biased. This is known informally as the ``winner's curse'': Large statistics tend to come from large parameters, but they also tend to be large by chance.

There is a related phenomenon for confidence intervals: Marginal confidence intervals almost always fail to cover parameters associated with the most significant estimates because they fail to account for the bias of these estimates. This results in overly short intervals that are too far from zero. 
Often this problem is recognized by investigators, but the most commonly used alternative, the Bonferroni correction, yields enormous, uninformative intervals. 
Problems with the Bonferroni correction have been described by \cite{Benjamini2005}, \cite{Zhong2008}, \cite{Weinstein2013} and several others. 
This method does not recenter the intervals and inflates their size symmetrically in order to control the family-wise error rate. Intuitively, we know that the largest estimates are more often too large than too small so most of the upper extension provided by the Bonferroni confidence intervals is unnecessary. Furthermore, the family-wise error rate is a much more conservative criterion than is typically desired.

We will discuss several alternative coverage criteria for high dimensional settings and introduce rank conditional coverage (RCC). Intuitively, RCC is the expected coverage probability of an interval given the ranking of its corresponding estimate. 
The RCC captures the idea that, when many parameters are estimated, the rank of an estimate provides information about its bias and the coverage probability of the associated confidence interval. 
This criterion applies to all high dimensional studies regardless of whether or not selection is performed but is particularly relevant when rank is used as a selection criterion. 
One advantage to obtaining confidence intervals that control the RCC, rather than selection adjusted confidence intervals, is that these intervals are not dependent on a particular selection procedure --- i.e. the interval calculated for a particular estimate will not change if the selection threshold is moved.
We discuss several approaches for obtaining useful confidence intervals --- these include a parametric and non-parametric bootstrap approach --- and explore their behavior and RCC in a few common settings.

\subsection{Coverage criteria after selection}\label{sec:coverage_criteria}
Consider an analysis in which we would like to obtain estimates and confidence intervals for many parameters $\theta_{1},\dots,\theta_{p}$. Suppose we have point estimates for each parameter $\hat{\theta}_{1}, \dots, \hat{\theta}_{p}$ and a \emph{marginally-valid} procedure for constructing confidence intervals. By marginally-valid we mean that for each parameter $\theta_j$, the coverage probability of the $\alpha$ level confidence interval $CI_j$ satisfies $P[\theta_j \in CI_j(\alpha)] = 1-\alpha$.

In the high dimensional setting, the marginal confidence interval will control the \emph{average coverage} --- i.e. if we construct 90\% marginal confidence intervals, we can expect them to cover 90\% of the parameters. 
Typically, however, the entire set of parameter estimates is not of interest and only the most significant estimates are reported. %\textcolor{blue}{
\cite{Benjamini2005} demonstrate that, for common selection rules, the expected rate of coverage within a selected subset of parameters will be much lower than the desired level. The marginal confidence interval achieves its average coverage by under-covering parameters associated with the most extreme or significant estimates and over-covering more boring parameters.%}

One might instead want to control the conditional coverage probability: the average coverage rate of selected parameters. More formally, Let $S$ be a set of indices for selected parameters (note that $S$ will be stochastic). The conditional coverage probability is defined as
\begin{align}
\mbox{Conditional Coverage Probability} = \frac{1}{p}\sum_{j=1}^{p}P[ \theta_{j} \in CI_{j} \vert j \in S]
\end{align}

Unfortunately, it is not possible to construct confidence intervals guaranteed  to control this criterion under many common selection procedures. For example, if selection is based on excluding 0 from the confidence interval but $\theta_{1}=\dots = \theta_{p}=0$, the conditional coverage probability will be 0 for any set of confidence intervals.

As an alternative, \cite{Benjamini2005} propose the concept of the false coverage statement rate (FCR). The FCR measures the probability of making a false coverage statement. This concept is very similar to conditional coverage probability.  There are two differences: 1) it averages coverage over the selected set, and 2) if nothing is selected, it is counted as ``no false statements'' being made. \cite{Benjamini2005} also restrict application of FCR to intervals which are constructed after selection.
\begin{align}
\mbox{False Coverage Statement Rate} = & E[Q]\\
Q =& \begin{cases} \frac{\sum_{j\in S} 1_{\theta_{j} \not\in CI_{j}}}{\vert S \vert} \qquad & \vert S \vert > 0\\
0 \qquad & \vert S \vert =0
\end{cases}\nonumber
\end{align}

To control FCR, \cite{Benjamini2005} give a procedure which, like the Bonferroni procedure,
 symmetrically inflates the size of marginally-valid intervals for all parameters. When selection is based on parameter estimates exceeding a threshold, these are equivalent to coverage $(1- \frac{\vert S \vert \alpha}{p})$ marginal intervals. This uniform inflation can be excessive in some cases. For example, if one parameter is very large, it will nearly always be selected. Thus, there is no need to inflate that interval at all (a $1-\alpha$ marginal interval will still control FCR). 
Additionally, the fact that highest ranked estimates are more often too large than too small suggests that  confidence intervals should have longer tails extending towards the bulk of the estimates than extending away from the bulk. 
 
%Additionally, intervals generated by this procedure 
%are still more likely to cover parameters associated with the least significant  estimates than those associated with the most significant estimates. are more likely to be over-estimates than under-estimates which suggests that confidence intervals should have longer tails extending towards the bulk of the estimates than extending away from the bulk. 

These issues are partially accounted for in more recent literature: \cite{Zhong2008}, \cite{Weinstein2013}, and \cite{Reid2014} all propose FCR controlling intervals which return to the marginal interval for very large parameter estimates. 
\cite{Zhong2008} and \cite{Weinstein2013} both condition on selecting all estimates larger than a (possibly data-dependent) cutoff. 
\cite{Zhong2008} use a likelihood-based approach to obtain asymptotically correct FCR while \cite{Weinstein2013} calculate exact intervals under the assumption that parameter estimates are independent with a known symmetric unimodal distribution. 
\cite{Reid2014} condition on the identity of the selected set and construct exact intervals for finite sample sizes assuming that parameter estimates are drawn from independent Gaussian distributions. All three of these intervals are asymmetric about the original point estimate.

Despite these advances, controlling FCR remains an unsatisfying solution to confidence interval construction for high dimensional problems. FCR controlling methods achieve the correct coverage rate within a subset in the same way that the marginal intervals achieve the correct rate in the larger set --- with under-coverage of the (more interesting) highest ranked parameters and over-coverage of (less interesting) more moderately ranked parameters. We illustrate this pattern through a simple example in Section~\ref{sec:example}.

\subsection{Rank conditional coverage}
\label{sec:rcc_intro}

In Section~\ref{sec:coverage_criteria} we observed that, for unadjusted confidence intervals as well as the selection adjusted alternatives, the rank of an estimate is informative about the probability that the associated confidence interval covers its target (see also Section~\ref{sec:example}). 
We find this phenomenon undesirable since it means that parameters associated with top ranked estimates are covered at a much lower rate than parameters associated with less significant estimates. 
Additionally, this observation indicates that there is an opportunity to use more information and construct better intervals. 

We first introduce the concept of rank conditional coverage (RCC) as a way to quantify the relationship between rank and coverage probability. 
%\textcolor{blue}{
In the majority of cases, the most interesting ranking of parameters is based on either the size of an associated test statistic or a $p$-value. 
In general, we assume that we have some ranking function $s$ where $s(i)$ gives the index of the $i$th ranked estimate. For example, if we are ranking simply based on the size of estimates, then
\[
\hat{\theta}_{s(1)} \geq \dots \geq \hat{\theta}_{s(p)}.
\]
In this paper, we will use the convention that a smaller rank indicates that an estimate is more significant, so the most significant estimate will have rank 1. In our examples, we focus on simple common rankings but the RCC could be defined for any scheme, and in fact the ranking scheme need not give a rank to every estimate. For example, if the parameter estimates can be grouped into highly correlated subsets such as LD blocks, we might choose the most significant estimate in each block and rank only this selected set. This type of ranking scheme is discussed at greater length along with simulation results in Section~\ref{sec:sims_block} of the Appendix.%}

We define the RCC at rank $i$ of a set of confidence intervals $CI_{1} \dots CI_{p}$ as
\begin{align}
\text{Rank Conditional Coverage}_{i} =&  P[\theta_{s(i)} \in CI_{s(i)}]\\
= & \sum_{j=1}^{p}P[ \theta_{j} \in CI_{j} | s(i) = j ]\cdot P[s(i)=j]
\end{align} 
This quantifies how often the interval formed around the $k$-th ranked estimate contains its target parameter. This is an appealing criterion, since we have a strong interest in ensuring that intervals around our most promising candidate features contain their targets. Something to note here is that we are not conditioning on which specific features achieve a given rank. Rather, we are averaging over all features (weighted by their probability of achieving that rank). While FCR summarizes the average coverage of a confidence interval procedure applied to a set of selected parameters in a single number, RCC gives a separate estimate of coverage probability for each rank and is not directly related to a selection procedure. 

\subsection{Implications of Controlling RCC}

Intervals that control RCC do not provide guarantees for particular parameters. For example, suppose $\theta_1$ is of special interest. If we use an RCC controlling method with $\alpha=0.1$, we cannot say that, if the experiment were repeated many times, $\theta_1$ would be contained in $CI_1$ in 90\% of experiments. 
Thus, if there is particular prior interest on one or a few parameters, RCC is not the correct criterion to control.

A correct statement that could be made about intervals controlling the RCC is that, if the experiment were repeated many times, 
 we expect the parameter corresponding to the top ranked estimate to be contained in its interval in 90\% of experiments.
 A similar statement could be made for any rank.
While this property may seem less intuitive on its surface, it has important implications when parameters are selected based on rank or significance. 

%\textcolor{blue}{
For example, suppose that a researcher publishes the results of many genome-wide association studies, each time reporting the most significant effect size estimates. If these estimates were paired with confidence intervals controlling the RCC at 90\%, the researcher could expect that 90\% of the published intervals for 1st ranked estimates (or 2nd etc.) averaging over studies contain their parameters. 
Most followup studies are conducted specifically for the most promising parameters, so this is precisely the type of guarantee needed to ensure these followup studies are worthwhile.
%}

%\textcolor{blue}{
This guarantee is stronger than that made by the FCR. In fact, it is straightforward to see that confidence intervals which control RCC for every rank also control FCR for selection rules that choose the most significant parameters based on the same ranking used to define RCC (see Theorem~\ref{thm:rcc_fcr} of the Appendix).
%}

%\textcolor{blue}{
One major strength of the RCC over the FCR is that RCC controlling confidence intervals can be divorced from the selection procedure. For example, if FCR controlling intervals are published for the top 10 parameter estimates but we are only able to follow up on the top 5 then we will need to recompute new, wider intervals in order to guarantee coverage within the smaller set. Using RCC controlling intervals, the same intervals remain valid regardless of how many parameters are selected.
%}
%However, the guarantee made by RCC is stronger since the intervals can be divorced from the selection procedure. 
%The researcher will also be able to expect that 90\% confidence intervals of the published 1st ranked estimate (or 2nd etc.) averaging over studies contain their parameter regardless of how many parameters were selected in each study. We will see that, using an FCR controlling interval construction procedure that does not control the RCC, the researcher can actually expect that the almost none of the intervals for published 1st ranked estimates contain their target parameter.

\subsection{Relationship of RCC to Empirical Bayes Approaches}
%{\color{blue} 
The observation that motivates the RCC is that, in a study estimating many parameters, the full set of estimates can provide information about the true underlying parameter values.
This is the same idea that motivates empirical Bayes (EB) approaches to simultaneous inference problems.
In a Bayesian paradigm, we are interested in estimating the posterior distributions of $\theta_1, \dots, \theta_p$ which, assuming conditional independence of the estimates and using Bayes rule, we can express as
\[
p(\theta_i | \hat{\theta}_i) \propto p(\hat{\theta}_i \vert \theta_i)p(\theta_i).
\]
The idea of EB approaches such as those of \cite{Efron2008} and \cite{Stephens2016} is to assume a theoretical distribution for $\hat{\theta}_i \vert \theta_i$ and use the large number of parameter estimates to estimate the prior $p(\theta_i)$.
For example, in the ashr method proposed by \cite{Stephens2016}, $p(\theta_i)$ is assumed to be unimodal and centered at zero and, in one of several proposed variations, is estimated as a mixture of normal distributions. Provided the EB modeling assumptions hold, 
%(e.g. the estimates are independent conditional on the true parameter values, the theoretic distribution of $\hat{\theta}_i \vert \theta_i$ holds approximately and any assumptions about the prior distribution of $\theta_i$ are true) 
we can expect that, averaging over many realizations, the $1-\alpha$ EB credible intervals contain the true parameter $1-\alpha \%$ of the time and are immune to selection bias. That is, in a Bayesian system where a single realization of an experiment also includes resampling the parameter values, EB credible intervals should control the RCC.
  %The common ground between the rank conditional coverhage criterion and the EB paradigm is the idea that a massively parallel experimental structure provides information about the distribution of true parameter values. 

The bootstrapping approaches we describe in Section~\ref{sec:par} differ from EB methods in that they require fewer modeling assumptions and are derived from a frequentist perspective. We tend to view the parameters $\theta_i$ as fixed 
and use the large number of estimates to learn about the distribution of the bias $\hat{\theta}_{s(i)} -\theta_{s(i)}$. This method requires no assumptions about the form of $p(\theta_i)$. In the non-parametric version, we are also able to avoid assumptions about the form of $p(\hat{\theta}_i | \theta_i)$,  provided we have access to the individual level data used to produce the original estimates. The flexibility of the non-parametric method does come at the expense of increased computational effort. The parameteric bootstrap can be fairly efficient and in the example in Section~\ref{sec:example} is 8 times faster than ashr. The non-parametric bootstrap can be quite costly since we must be willing to repeat the entire analysis hundreds of times. However, in cases in which the parameter estimates are not independent or the theoretical distributions of test statistics are poor approximations, the non-parametric bootstrap is the most appropriate choice. 
%}

\subsection{Example}
\label{sec:example}
Consider 1000 independent estimates $Z_{i}$ is drawn from a  $N(\theta_{i}, 1)$ distribution and ranked according to their absolute value. 
Figure \ref{fig:example} shows the average RCC over 100 simulations for the top 20\% of statistics for several different configurations of parameter values:
\begin{enumerate}[1.]
\item All parameters equal to zero.
\item All parameters small and non-zero: $\theta_{i}$ generated from a $N(0, 1)$ distribution but fixed for all simulations.
\item A few large non-zero parameters: $\theta_{i}=3$ for $i = 1\dots 100$ and $\theta_{i} = 0$ for $i > 100$.
\item A few small non-zero parameters: $\theta_{i}$ drawn from a $N(0, 1)$ distribution but fixed over all simulations for $i = 1\dots 100$ and $\theta_{i} = 0$ for $i > 100$.
\end{enumerate}
The standard marginal confidence intervals are $CI_{i} = Z_{i} \pm \Phi^{-1}(1-\alpha/2)$. 
In configuration 1, this interval has an RCC of $\sim 0$\% for the 65 most extreme observations but an RCC of $\sim 100$\% for statistics closer to the median giving an overall average of 90\% coverage.

\begin{figure}
\centering

Average Rank Conditional Coverage\\
\includegraphics[width=\textwidth]{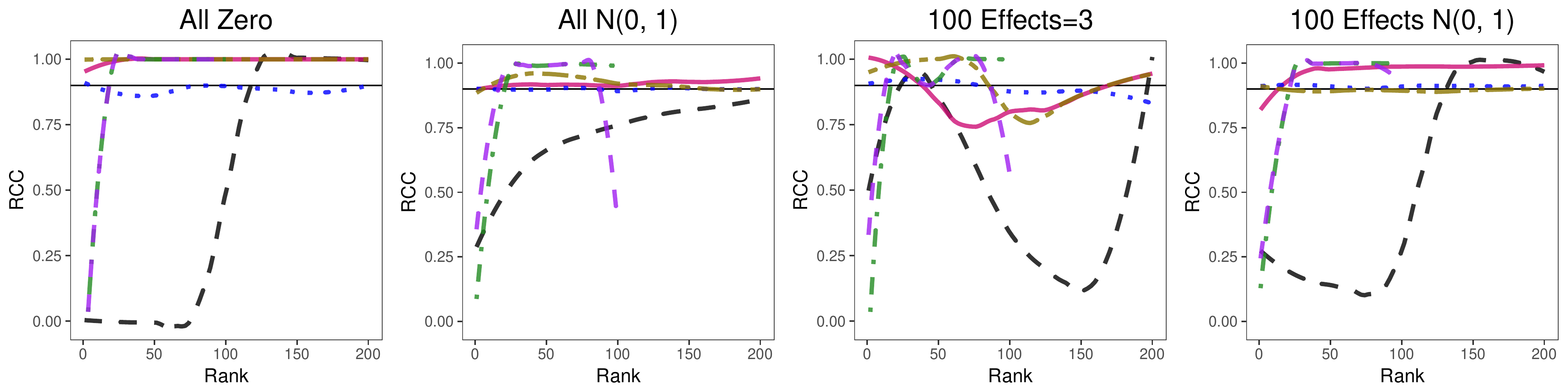}\\
Interval Width\\
\includegraphics[width=\textwidth]{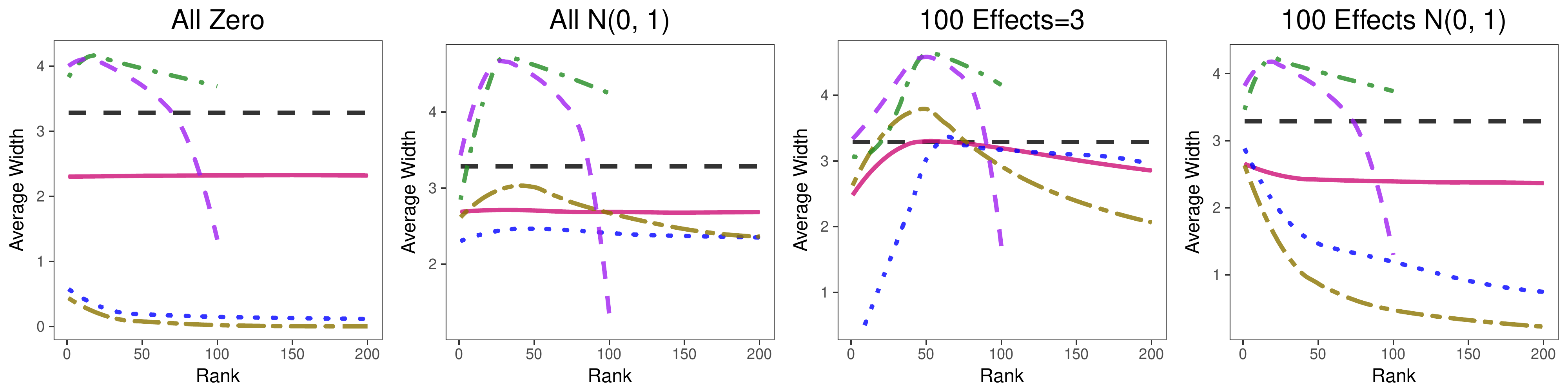}\\
\includegraphics[width=\textwidth]{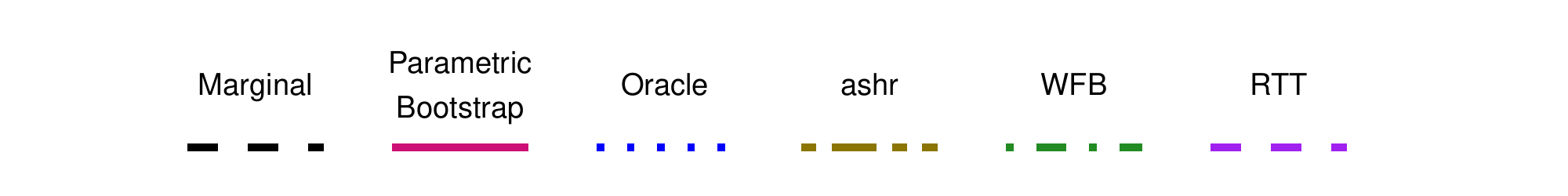}

\caption{Average rank conditional coverage (top) and interval widths (bottom) for the example described in Section~\ref{sec:example}. The top 20\% of statistics are shown and both coverage and width are smoothed using loess. The four sets of true parameters are described in Section~\ref{sec:example}. 
The horizontal line in the top plots shows the nominal level 90\%.
}
\label{fig:example}
\end{figure}

We see a similar pattern in the intervals constructed by \cite{Weinstein2013} and \cite{Reid2014}. Both provide intervals only for a selected subset of parameters (we selected the most extreme 10\% of observations). Both methods control FCR but do so by under-covering parameters associated with the most significant statistics and over-covering parameters with more moderate statistics. In settings 2 and 3, the intervals of \cite{Reid2014} have poor RCC for both the most and least extreme parameters considered. 
%\textcolor{blue}{
The credible intervals generated by the \texttt{ashr} method of \cite{Stephens2016} do control the RCC in this setting and can sometimes be very small. \texttt{ashr} is able to achieve a very small average interval width in settings 1, 3, and 4 because it attempts to shrink parameter estimates to zero. If the posterior probability that the parameter is equal to zero is larger than the desired level, the resulting credible interval will simply be $[0, 0]$.
%}

Figure~\ref{fig:example} also shows the results of the parametric bootstrapping method  described in Section~\ref{sec:par}. This method provides an RCC close to the nominal level for all ranks and, in most cases, a shorter interval.
We see under-coverage for a small set of parameters in scenarios 3 and 4 in which most parameters are zero and a handful take on larger values. In all four scenarios, the bootstrapped confidence intervals are also shorter than either of the selection adjusted methods and sometimes shorter than the ashr credible intervals.

The deviations from the nominal level of RCC are a result of using estimated mean values to generate bootstrap samples. Hypothetically, if these values were known, we could produce an ``oracle'' estimate which, with enough Monte-Carlo samples, would achieve exactly the desired RCC for all order statistics. The oracle is shown in Figure~\ref{fig:example} and provides the motivation for the methods described in Section~\ref{sec:par} where it is discussed in more detail.
%\textcolor{blue}{
An in depth walk-through of these results as well as code for replicating Figure~\ref{fig:example} is available in 
\url{https://github.com/jean997/rccSims/walkthroughs/compare_cis.pdf}.

We explore the performance of these methods in two simulation studies designed to mimic common high dimensional analyses in Section~\ref{sec:sim}. %In the first, genotype data and phenotypes are simulated and effects are estimated through linear regression. %In the second, we estimate pairwise correlations between simulated gene expression data which is a common task in the construction of coexpression networks.
%{\color{blue}
In the first, we simulate an outcome and many features and estimate marginal effect sizes via single-variable linear regression. This is a common approach to genome-wide association studies.
In the second, we estimate the difference in average treatment effect for subsets of individuals defined the level of a biomarker. These estimates are highly correlated making the non-parametric bootstrap the most appropriate method.
%}
\section{Parametric and Non-Parametric Bootstrapping to Build Confidence Intervals}
\label{sec:par}
\subsection{Rank conditional confidence intervals}
First consider estimating a single parameter of a single distribution $\theta = T(F)$. Let $\delta = \hat{\theta}-\theta$ and $H(x) = P[\delta \leq x]$ be the cdf of $\delta$. When $H$ is known, a pivotal exact $1-\alpha$ confidence interval can be constructed as
\begin{align}
\Big( \hat{\theta} - H^{-1}(1-\alpha/2), \ \hat{\theta} - H^{-1}(\alpha/2)\Big) \label{eq:ci_classic}
\end{align}

%%%% Talk about how we get short intervals! and make a graphic!

In the high dimensional setting, we are attempting to estimate $p$ parameters $\theta_{i} = T_{i}(F)$. We can construct a rank conditional analog of the classical pivotal interval in (\ref{eq:ci_classic}). Define $\hat{\theta}_{s(i)}$ as in section \ref{sec:rcc_intro} where $s(i)$ gives the index of the $i$th ranked parameter estimate. We define the bias of the estimates at each rank
\begin{align}
\delta_{[i]} = \hat{\theta}_{s(i)}-\theta_{s(i)}
\end{align}
where the subscript $[i]$ indicates ranked based indexing.
Let $H_{[i]} = P[\delta_{[i]} \leq x]$ be the cdf of $\delta_{[i]}$. Were $H_{[i]}$ known, an exact $1-\alpha$ confidence interval for $\theta_{s(i)}$ could be constructed as 
\begin{align}
CI_{s(i)}^{\text{exact}} = \Big(\hat{\theta}_{(i)}-H_{[i]}^{-1}(1-\alpha/2),\ \hat{\theta}_{(i)} -H^{-1}_{[i]}(\alpha/2)\Big) \label{eq:exactci}
\end{align}
We note that the rank conditional intervals are not pivotal because the distribution of $\delta_{[i]}$ depends on $\theta_{1}, \dots, \theta_p$. This makes them more difficult to obtain when $H_{[i]}$ are unknown but doesn't impact the coverage probability of \eqref{eq:exactci}.

\begin{lem}
The intervals in \eqref{eq:exactci} have exact $1-\alpha$ coverage:
\begin{equation*}
P[\theta_{s(i)} \in CI_{s(i)}^{\text{exact}}] = 1-\alpha
\end{equation*}
\end{lem}
\begin{proof}
This proof is identical to the proof for the classical interval in (\ref{eq:ci_classic}) given by  \cite{Wasserman2006} among others.
Let $a = \hat{\theta}_{s(i)}-H_{[i]}^{-1}(1-\alpha/2)$ and $b = \hat{\theta}_{s(i)}-H_{[i]}^{-1}(\alpha/2)$. 
\begin{eqnarray*}
P[ a \leq \theta_{s(i)} \leq b] =& P[\hat{\theta}_{s(i)} - b \leq \delta_{[i]} \leq \hat{\theta}_{s(i)} -a]\\
=& H_{[i]}(\hat{\theta}_{s(i)}-a) - H_{[i]}(\hat{\theta}_{s(i)} - b)\\
=& H_{[i]}\Big( H_{[i]}^{-1}(1-\alpha/2) \Big) - H_{[i]}\Big(H_{[i]}^{-1}(\alpha/2)\Big)\\
=& 1-\frac{\alpha}{2}-\frac{\alpha}{2} = 1- \alpha
\end{eqnarray*}
\end{proof}

%\textcolor{red}{I would talk about two things in this section: (1) that the oracle intervals are very short compared to the bonferroni intervals and are potentially even shorter than the unadjusted intervals. (2) Might there be some graphic we can make to make these things more clear? Perhaps another subsection for these?}

\subsection{Generating oracle intervals with Monte Carlo sampling}
Construction of the intervals in \eqref{eq:exactci} requires knowledge of the quantiles of the cdfs $H_{[i]}$. Working directly with $H_{[i]}$ may be difficult. If, instead, we can easily sample from the joint distribution $G$ of $\hat{\vs{\theta}} = (\hat{\theta}_1, \ldots, \hat{\theta}_{p})^{\top}$, then the quantiles of $H_{[i]}$ can be computed via Monte-Carlo. 
We now describe this oracle Monte-Carlo procedure which is detailed in algorithm \ref{alg:oracle} and illustrated in Figure~\ref{fig:algoracle}. 

First, we draw $K$ independent $p$-vectors $\vs{\vartheta}_1 \dots \vs{\vartheta}_K$ from $G$. Let $\vartheta_{k,s_{k}(i)}$ be the $i$th %{\color{blue} 
ranked%} 
element of $\vs{\vartheta}_{k}$ and $s_k$ be the ranking permutation function for $\vs{\vartheta}_k$. Define the observed bias in sample $k$ at rank $i$ as
\begin{equation}
\tilde{\delta}_{k,[i]} = \vartheta_{k,s_k(i)} - \theta_{s_k(i)} \label{eq:delta}
\end{equation}
In the left panels of Figures~\ref{fig:oracle_1} and \ref{fig:oracle_100}, the true distribution $G$ is shown in the background and the wider distributions of three sets of Monte-Carlo replicates are overlayed. In Figure~\ref{fig:oracle_1}, solid lines mark $\vartheta_{k, (1)}$ for $k = 1, 2, 3$ and dashed lines mark the corresponding true parameter value $\theta_{s_k(1)}$. The difference between these lines is the bias $\tilde{\delta}^{k, [1]}$. In Figure~\ref{fig:oracle_100}, the same lines are shown for the 100th rank.

Next, for each rank, we use the sample quantiles of $\left\{\tilde{\delta}_{1, [i]}\dots \tilde{\delta}_{K, [i]}\right\}$ to estimate the quantiles of $H_{[i]}$. We denote these sample quantiles as $\tilde{H}^{-1}_{[i]}(\cdot)$. 
In the middle panels of Figures~\ref{fig:oracle_1} and \ref{fig:oracle_100}, we show the distributions of $\left\{\tilde{\delta}_{1, [1]}\dots \tilde{\delta}_{K, [1]}\right\}$ and $\left\{\tilde{\delta}_{1, [100]}\dots \tilde{\delta}_{K, [100]}\right\}$ and the 0.05 and 0.95 quantiles are marked by vertical dashed lines. Note that the distribution of bias for the largest statistic is right shifted compared to the distribution of bias for the 100th statistic. 
 
Substituting these estimates into the interval in \eqref{eq:exactci} gives the oracle confidence interval
\begin{equation}
CI_{s(i)}^{\text{oracle}} = \Big(\hat{\theta}_{s(i)}-\tilde{H}^{-1}_{[i]}(1-\alpha/2),\ \hat{\theta}_{s(i)} -\tilde{H}^{-1}_{[i]}(\alpha/2)\Big) \label{eq:ci}
\end{equation}
This pivot is illustrated in the right panels of Figures~\ref{fig:oracle_1} and \ref{fig:oracle_100}. For both of the ranks shown, the oracle Monte-Carlo intervals are shorter than the marginal interval and contain the true parameter value.

These intervals are called oracle confidence intervals because they use knowledge of $G$ and $\theta_{1}, \dots, \theta_{p}$. We know that $\hat{H}^{-1}_{[i]}(x)\rightarrow H^{-1}_{[i]}(x)$ as the number of Monte-Carlo samples increases, so using (\ref{eq:ci}) we can achieve the correct $1-\alpha$ confidence level (within any $\epsilon$ tolerance). This can be seen in Figure~\ref{fig:example} where the oracle intervals have very close to the target 90\% rank conditional coverage at all ranks and are shorter than other methods. The following sections describe bootstrapping methods for estimating $H_{[i]}^{-1}$ when the distribution of $\hat{\vs{\theta}}$ is unknown.

\begin{figure}
\centering
	\begin{subfigure}[b]{\textwidth}
		\caption{Oracle interval construction for the parameter with the largest estimate.}
		\includegraphics[width=\textwidth]{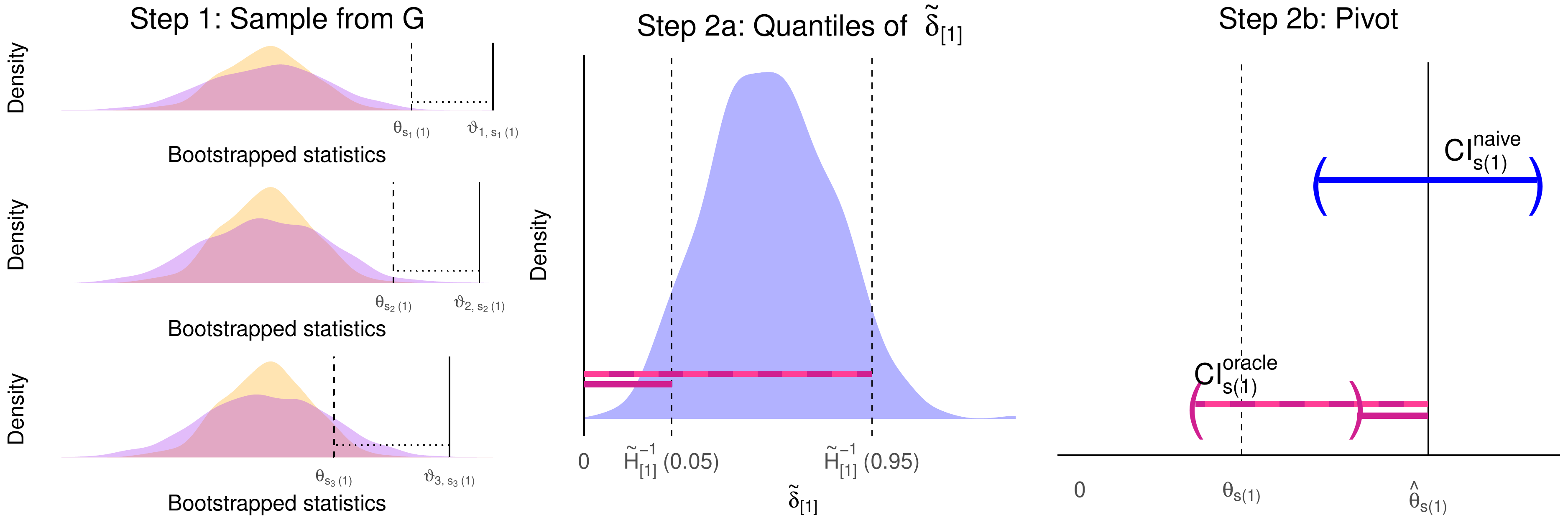}
		\label{fig:oracle_1}
	\end{subfigure}
	\begin{subfigure}[b]{\textwidth}
	\caption{Oracle interval construction for the parameter with the 100th largest estimate.}
		\includegraphics[width=\textwidth]{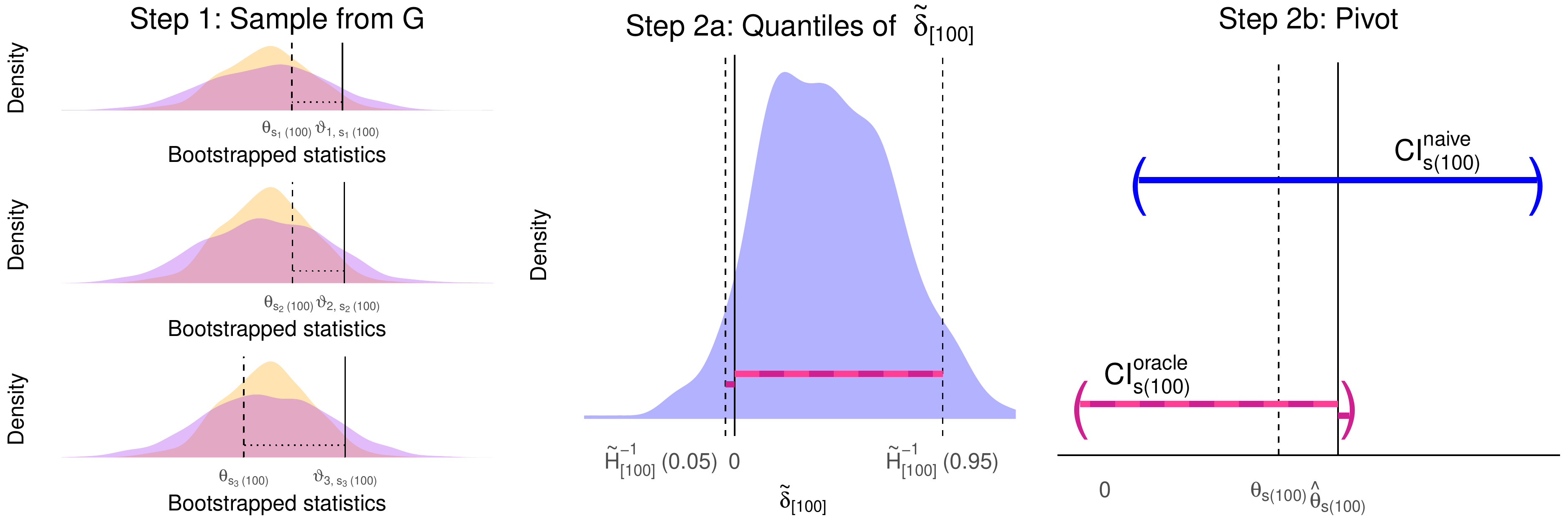}
		\label{fig:oracle_100}
	\end{subfigure}

\caption{Generating oracle confidence intervals using algorithm \ref{alg:oracle} for rank $i=1$ in \ref{fig:oracle_1} and $i=100$ in \ref{fig:oracle_100}.
\emph{Left panels:} The distribution of the true parameters, $G$, is shown in the background. The wider distributions from three sets of Monte Carlo replicates are overlayed. Solid vertical lines mark the locations of the largest (Figure~\ref{fig:oracle_1}) and 100th largest (Figure~\ref{fig:oracle_100}) element of the Monte Carlo sample. Dashed  vertical lines mark the corresponding true parameter value. The distance between these lines is the bias, $\delta_{k, [i]}$.
\emph{Middle:} The distribution of biases for 500 Monte Carlo samples with 0.05 and 0.95 quantiles marked by vertical dashed lines. 
\emph{Right:} The oracle interval is constructed by pivoting the quantiles in the middle panel around the observed test statistic $\hat{\theta}_{s(i)}$. The solid and dashed horizontal lines extending from $\hat{\theta}_{s(i)}$ are the same length as the solid and dashed horizontal lines in the middle panel. The dashed vertical line shows the location of the true parameter. The naive interval is shown for comparison. The vertical axis is meaningless.
}
\label{fig:algoracle}
\end{figure}

\begin{algorithm}[h!]
\caption{Generating oracle intervals}\label{alg:oracle}
$G$ and $\theta_1, \dots , \theta_p$ are known.
\begin{enumerate}[1.]
\item For $k$ in $1 \dots K$:
\begin{enumerate}[a.]
\item Sample $(\vartheta_{k, 1}\dots \vartheta_{k,p})$ from $G$.
\item Calculate the bias at each rank $\tilde{\delta}_{k,[i]}$ as in (\ref{eq:delta}).
\end{enumerate}
\item For $i$ in $1\dots p$
\begin{enumerate}[a.]
\item Calculate empirical quantiles $\tilde{H}^{-1}_{[i]}(x)$ of $\lbrace \tilde{\delta}_{1,[i]}\dots \tilde{\delta}_{K,[i]} \rbrace$
\item Generate $CI_{s(i)}^{\text{oracle}}$ as in \eqref{eq:exactci}
\end{enumerate}
\end{enumerate}
\end{algorithm}

The intervals, $CI_{s(i)}^{\text{oracle}}$, in  \eqref{eq:ci} do not necessarily contain $\hat{\theta}_{s(i)}$. This is particularly true if the point-estimates, $\hat{\theta}_{s(i)}$, have not been adjusted for multiplicity/selection bias --- for example, if each $\hat{\theta}_{s(i)}$ is a maximum likelihood estimate. 
In Figure~\ref{fig:oracle_1}, both the $1-\alpha/2$ and $\alpha/2$ quantiles of the observed-bias are positive so the confidence interval lies completely below $\hat{\theta}_{s(1)}$.
%for the most significant estimate will often both be positive so the generated confidence interval will lie completely below the unadjusted estimate $\hat{\theta}_{(1)}$. 
A more natural point estimate to pair with this confidence interval is the de-biased estimate proposed by \cite{Simon2013}, $\hat{\theta}_{s(i)} - \frac{1}{K}\sum_{k=1}^{K} \tilde{\delta}_{k,[i]}$, which will generally lie within the confidence interval. %\textcolor{red}{My additions here might be a bit clunky. I think it is important though to note that the above depends on what we use for $\hat{\theta}_{(i)}.$}

The following bootstrap methods for confidence interval construction are extensions of point estimation methods proposed by \cite{Simon2013} and \cite{Tan2014}. Those former are based on estimating the mean of $\delta_{[i]}$ while we estimate quantiles. We consider two bootstrapping strategies --- a parametric bootstrap, useful when the distribution and covariance of the parameter estimates are known or can be approximated well, and a non-parametric bootstrap method which can be applied to any set of estimates based on iid samples, but is more computationally costly. Both of these are general strategies, where the specifics of the algorithm may vary depending on the specific application.
\subsection{Parametric Bootstrap}
\label{sec:par_boot}
The parametric bootstrap parallels the Monte Carlo algorithm in Algorithm~\ref{alg:oracle} 
replacing $G$ and $\theta_{1}, \dots, \theta_p$ with estimates based on the data. 

We assume that $G$ is a member of a parametric family of distributions and estimate its parameters. 
Though, in principal we could use any family of distributions, in most cases we will wish to assume that $\hat{\theta}_{i} \sim N(\theta_{i}, \sigma_{i}^{2})$ where $\sigma_{i}^{2}$ is either known or can be estimated and $\hat{\theta}_{i}$ are independent. This type of parametric bootstrap is best suited for scenarios in which the estimator has an asymptotically normal distribution with known variance such as linear regression. $\hat{G}$ can be chosen by replacing $\theta_{i}$ with an estimate such as $\hat{\theta}_{i}$ itself or de-biased estimates of \cite{Simon2013} or \cite{Tan2014}. The latter choices will involve two stages of bootstrapping --- the first to generate a de-biased mean estimate and the second to generate confidence intervals.

The quantiles estimated through Monte Carlo simulation in \eqref{eq:exactci} are replaced by bootstrapped quantiles $\hat{H}^{-1}_{[i]}(x)$ obtained by sampling $p$-vectors from $\hat{G}$ rather than from $G$. This gives the bootstrap intervals
\begin{align}
CI_{s(i)}^{\text{boot}} = \Big(\hat{\theta}_{s(i)}-\hat{H}^{-1}_{[i]}(1-\alpha/2),\ \hat{\theta}_{s(i)} -\hat{H}^{-1}_{[i]}(\alpha/2)\Big) \label{eq:boot_ci}
\end{align}
This procedure is described in algorithm \ref{alg:parametric} for the case when $G = N(\vs{\theta}, \v{I})$ where $\vs\theta = (\theta_1, \dots, \theta_p)$ and $\v{I}$ is the $p\times p$ identity matrix. Many variations on this procedure are possible. For example, it may be easier to specify a distribution for a transformation of $\theta_i$.% as in the simulations in section \ref{sec:sims_corr}.
When the family assumptions on $G$ are correct or asymptotically correct, a better RCC is obtained when the estimates used to calculate the bias are closer to the true parameter values. Algorithm \ref{alg:parametric2} of the Appendix shows the additional steps necessary when using a debiased mean estimate. 
If the ranking scheme is based on the absolute value of the parameter estimates (for example, using the magnitude of a $t$-statistic) it is necessary to reflect the the pivot across zero for negative parameter estimates. 
Algorithm \ref{alg:absparametric} of the Appendix gives the parametric bootstrap procedure for absolute value based rankings. %\textcolor{blue}{
Algorithm~\ref{alg:parametric} and the two variations described in the Apppendix are implemented in the \texttt{par\_bs\_ci} function of the R package \texttt{rcc}.%}

\begin{algorithm}[h!]
\caption{Simple parametric bootstrap for asymptotically normal estimates}\label{alg:parametric}
\begin{enumerate}[1.]
\item For $k$ in $1 \dots K$:
\begin{enumerate}[a.]
\item Sample $\vartheta_{k, i}$ from a $N(\hat{\theta}_{i}, 1)$ distribution 
for $i$ in $1\dots p$.
\item Calculate the bias at each rank $\hat{\delta}_{k,[i]}$ as 
\[\hat\delta_{k,[i]} = \vartheta_{k,s_k(i)} - \hat{\theta}_{s_k(i)}\]
\end{enumerate}
\item For $i$ in $1\dots p$:
\begin{enumerate}[a.]
\item Calculate empirical quantiles $\hat{H}^{-1}_{[i]}(x)$ of $\lbrace\hat\delta_{1,[i]}\dots \hat\delta_{K,[i]}\rbrace$
\item Generate $CI_{s(i)}^{\text{boot}}$ as in (\ref{eq:boot_ci})
\end{enumerate}
\end{enumerate}
\end{algorithm}

\subsection{Non-Parametric Bootstrap}
\label{sec:npar_boot}
The parametric bootstrap can be applied when $\hat{\vs{\theta}}$ has a distribution that is well approximated by a member of a parametric family. It is particularly convenient for statistics which are asymptotically normal and either independent or with a covariance that can be estimated well. Many high dimensional problems possess complex dependence structures which are not easy to estimate. Furthermore, sometimes our estimators do not have a known distribution. In these cases, the parametric bootstrap, the selection adjusted FCR controlling methods, and EB methods that assume conditional independence between estimates are unsuitable. 
%\textcolor{blue}{
In general, it is not possible to estimate a general $G$ without making any structural assumptions about the true parameters. However,
it is possible to generate bootstrap samples non-parametrically if individual data are available.%}

%\textcolor{blue}{
The non-parametric bootstrap is based on sampling from the data used to compute $\hat{\vs{\theta}}$ and computing new estimates using the re-sampled data. This is implicitly sampling from a distribution $\hat{G}$ without requiring an analytical form.%} 
We assume the data consist of $n$ independent data vectors $\v{y}_{1}, \ldots, \v{y}_{n}$. These may be vectors of genotypes, biometric, or image data for $n$ individuals. They may be a mix of data types and include covariates. We assume only that there is a procedure which takes $\v{y}_{1},\ldots, \v{y}_{n}$ as inputs and generates estimates $\hat{\vs{\theta}}$  and statistics indicating the significance of each estimate. 

A bootstrap $p$-vector %of parameter estimates %$\vs\vartheta$ is 
can be generated by sampling $n$ data vectors from $\v{y}_{1},\ldots, \v{y}_{n}$ without replacement and applying the original estimation procedure.
From this point confidence intervals may be constructed identically to the parametric case. 

More formally, if $\v{y}_{1},\ldots, \v{y}_{n}$ are iid draws from $\Pi$, and $\hat{\vs\theta}\equiv \hat{\vs\theta}\left(\v{y}_{1},\ldots, \v{y}_{n}\right)$ is a function of those observations, then $G\equiv G\left(\Pi\right)$ is directly a function of $\Pi$. 
To estimate $G$, we can use the estimate induced by the empirical distribution of the $\v{y}_i$: $\hat{G}_{emp} \equiv G\left(\Pi_n\right)$. From here we can estimate the quantiles in \eqref{eq:exactci} by $\hat{H}^{-1}_{emp[i]}(x)$ obtained by sampling repeatedly from $\hat{G}_{emp}$. This leads us to the non-parametric bootstrap intervals:
\begin{align}
CI_{s(i)}^{\text{np-boot}} = \Big(\hat{\theta}_{s(i)}-\hat{H}^{-1}_{emp[i]}(1-\alpha/2),\ \hat{\theta}_{s(i)} -\hat{H}^{-1}_{emp[i]}(\alpha/2)\Big) \label{eq:non-para-boot_ci}
\end{align}
The specifics of this procedure are shown in Algorithm~\ref{alg:nonparametric}  %\textcolor{blue}{
and implemented in the \texttt{nonpar\_bs\_ci} function of the R package \texttt{rcc}.%}
Non-parametric bootstrapping can potentially be very time consuming. If the original analysis was computationally expensive it may be infeasible to repeat it many times to obtain confidence intervals.

\begin{algorithm}[h!]
\caption{Non-parametric bootstrap}\label{alg:nonparametric}
\begin{enumerate}[1.]
\item For $k$ in $1 \dots K$:
\begin{enumerate}[a.]
\item Sample $\v{y}_{k,1}, \ldots \v{y}_{k,n}$ with replacement from $\lbrace \v{y}_{i}\rbrace$
\item Using the sampled data, calculate estimates $\vartheta_{k,1}\dots \vartheta_{k,p}$.
\item Estimate the bias at each rank $\hat{\delta}_{k,[i]}$ as 
\[
\hat\delta_{k,[i]} = \vartheta_{k,s_k(i)} - \hat{\theta}_{s_k(i)}
\]
\end{enumerate}
\item For $i$ in $1\dots p$:
\begin{enumerate}[a.]
\item Calculate empirical quantiles $\hat{H}^{-1}_{emp[i]}(x)$ of $\lbrace\hat\delta_{1,[i]}\dots \delta_{K,[i]}\rbrace$
\item Generate $CI_{s(i)}^{\text{np-boot}}$ as in (\ref{eq:non-para-boot_ci})
\end{enumerate}
\end{enumerate}
\end{algorithm}

\section{Simulations}\label{sec:sim}

%{\color{blue}
\subsection{Linear Regression with Correlated Features}
\label{sec:sims_lin}
In this set of simulations, we explore how correlation among parameter estimates effects the rank conditional coverage rates of different methods of confidence interval construction.
Code replicating these results as well can be found at \url{https://github.com/jean997/rccSims/walkthroughs/linreg_sims.pdf}.

We consider a common analysis procedure used in genetic and genomic studies. In these studies, researchers measure far more features (such as gene expression levels) than there are samples. They therefore focus on estimating the marginal association between each feature and an outcome. 
We consider a setting wherein the features occur in correlated blocks leading to correlated parameter estimates.

In each simulation, we simulate 1000 normally distributed features features for 100 samples. Let $x_{i,j}$ denote the value of the $j$th feature for the $i$th individual and $\v{x}_i = (x_{i,1}, \dots, x_{i,1000})^{T}$. The features are simulated as 
\[
\v{x}_i \sim N_{1000}(0, \Sigma)
\]
where the covariance matrix, $\Sigma$, is block diagonal with 100 $10\times 10$ blocks. The diagonal elements of each block are equal to 1 and the off diagonal elements are equal to $\rho$.  The outcome for individual $i$ is simulated as 
\[
y_i = \v{x}_i \vs{\beta} + \epsilon_i\qquad \epsilon_i \sim N(0, 1)
\]
where elements of $\vs{\beta}$ (the vector conditional effect sizes) are equal 0 at all but 100 elements. In each block the effect size for the 5th feature is drawn from a $N(0, 1)$ distribution while the effects for the other features are 0. These effects are fixed over all simulations. 

In this analysis we estimate the marginal rather than conditional effect sizes, $\beta^{(marg)}~=~\Sigma \vs{\beta}$. We estimate $\beta_j^{(marg)}$ through a univariate linear regression of $\v{x}_j = (x_{1, j}, \dots , x_{n, j})^{T}$ on $\v{y}$. This is a standard analysis strategy for many genomic studies such as genome wide association studies and gene expression studies.

We consider four levels of correlation between the features by setting $\rho$ equal to $0, 0.3, 0.8$ and $-0.1$. Rank conditional coverage and interval widths averaged over 400 simulations for each scenario are shown in Figure~\ref{fig:linreg}. 
In these results, parameter ranking is based on the absolute value of the $t$-statistic $\hat{\beta}_j^{(marg)}/\hat{se}(\hat{\beta}_j)$. In Section~\ref{sec:sims_block} of the Appendix we consider ranking the parameters by first selecting the parameter with the most significant estimate in each block and then ranking only these 100 selected parameters based on the absolute value of the $t$-statistic.
 
We find that both the parametric and non-parametric confidence intervals perform well in all four settings and are quite similar, even though the parametric bootstrap assumes independence between the estimates. None of the other methods provides an RCC close to the nominal level except for ashr in the highest correlation scenario. The ashr method does poorly in the other scenarios because the marginal effects are not sparse. Since ashr attempts to shrink parameters to zero, the credible intervals will often be to close to zero and too small in settings when the true parameter values are not sparse. 
 
\begin{figure}
\centering
Average Rank Conditional Coverage\\
\includegraphics[width=\textwidth]{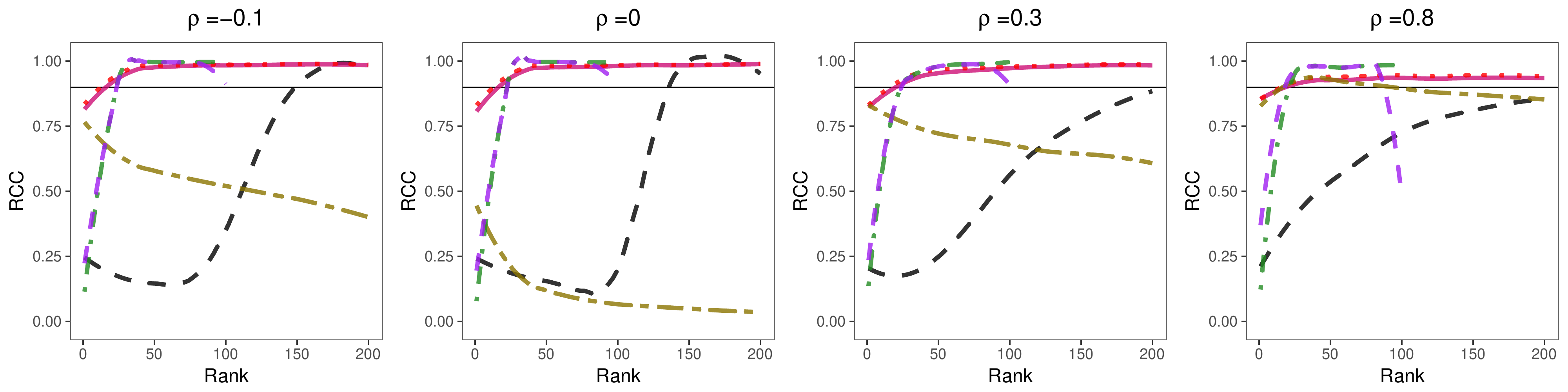}
Interval Width\\
\includegraphics[width=\textwidth]{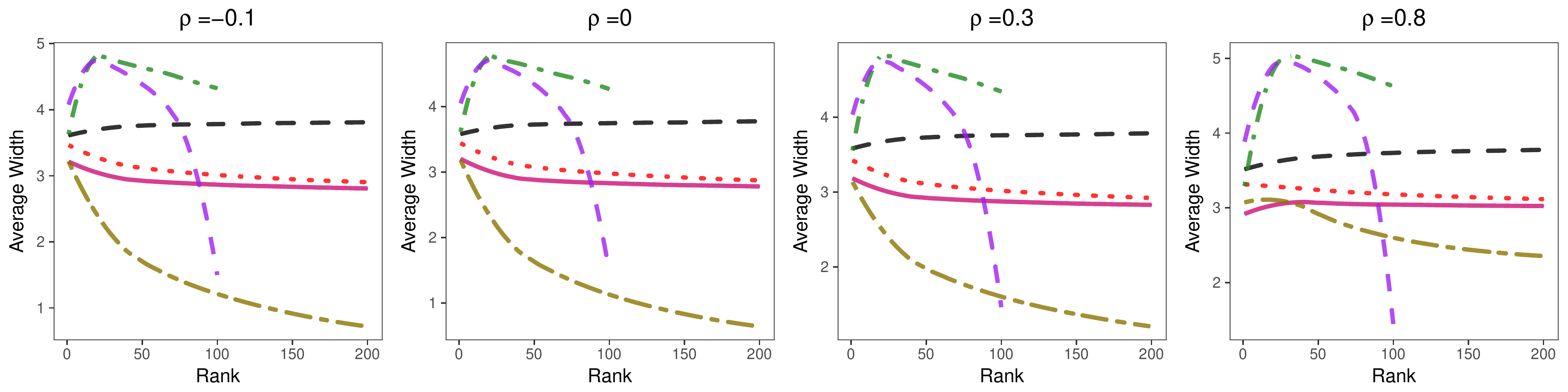}\\
\includegraphics[width=\textwidth]{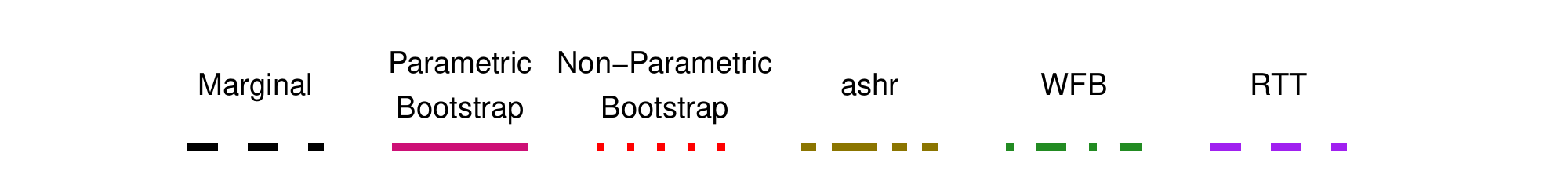}

\caption{Simulation results for Section~\ref{sec:sims_lin}. Rank conditional coverage (top) and interval widths (bottom) are shown for the top 20\% of parameters averaged over 400 simulations. Parameters are ranked by absolute value of the test statistic. Coverage rates and widths are smoothed using loess. 
In the top panel, a horizontal line shows the nominal level 90\%.
}
\label{fig:linreg}
\end{figure}

\subsection{Treatment Effects in Nested Subgroups}\label{sec:sims_biomarker}
In Section~\ref{sec:sims_lin}, we found that the parametric bootstrap performed well even when the assumption of independence between estimates was violated. 
Here we provide an example of how the parametric bootstrap can fail when estimates are very highly correlated. 
Code replicating these results can be found at \url{https://github.com/jean997/rccSims/walkthroughs/biomarker_sims.pdf}.

This example is motivated by the use of biomarkers in clinical trials. Suppose we have conducted a clinical trial in which participants are randomized into two groups. For participant $i$, we record the treatment group, $trt_i$, an outcome $y_i$ and the value of a biomarker, $w_i$. We expect that the treatment will have a greater effect in individuals with higher values of the biomarker but don't know the exact relationship between the biomarker and the treatment effect. In an exploratory analysis, we define a series of cut-points $c_1, \dots, c_p$. For each cut-point we estimate the difference in treatment effects for participants with biomarker measurements above and below the cut-point:
\begin{align*}
\beta_j = &  \left (E[y_i | trt_i=1, w_i > c_j] - E[y_i| trt_i=0, w_i > c_j] \right) -\\ 
  & \left (E[y_i | trt_i=1, w_i  \leq c_j] - E[y_i| trt_i=0, w_i \leq c_j] \right).
\end{align*}
We estimate $\beta_j$ as the OLS estimate fitting the regression 
\[
y_i = \beta_0 + \beta_1 trt + \beta_{2,j} 1_{w_i > c_j} + \beta_{j} 1_{w_i > c_j} + \epsilon_i,
\]
where $1_{w_i > c_j}$ is an indicator that $w_j > c_j$. 
We then rank these estimates by the absolute value of their $t$-statistics in order to select a cut-point that gives the most significant difference in treatment effect between groups. This cut-point might be used to design future clinical trials.

In each simulation, we generate data for 200 study participants, 100 randomized to the treatment arm and 100 randomized to the control arm. We simulate the value of the biomarker as uniformly distributed between 0 and 1. The true relationship between the biomarker, the treatment, and the outcome given by 
\[
E[y_i | trt_i, w_i] = \begin{cases}
 0 \qquad & w_i < 0.5\\
  (w_i-\frac{1}{2})\cdot trt_i & w_i \geq 0.5
\end{cases}.
\]
The observed outcome for individual $i$ ($i \in 1, \dots, 200$) is $y_i = E[y_i | w_i, trt_i] + \epsilon_i$ where $\epsilon_i \sim N(0, 0.25)$. 

We chose 100 cut-points evenly spaced between 0.1 and 0.9. Rank conditional coverage and interval width averaged over 400 simulations are shown in Figure~\ref{fig:biomarker}. In this scenario, parameter estimates are very highly correlated. This results in very poor performance for the parametric bootstrap which assumes independence between estimates. Interestingly, the standard marginal intervals do well despite making the same assumption. The non-parametric bootstrap also controls the RCC though it has slight under-coverage for the least significant parameters. Unlike the marginal intervals, the non-parametric bootstrap controls the RCC by modeling the correlation structure between parameter estimates and also performs well in the simulations in Section~\ref{sec:sims_lin} making it a more reliable choice.

\begin{figure}
\centering
\includegraphics[width=\textwidth]{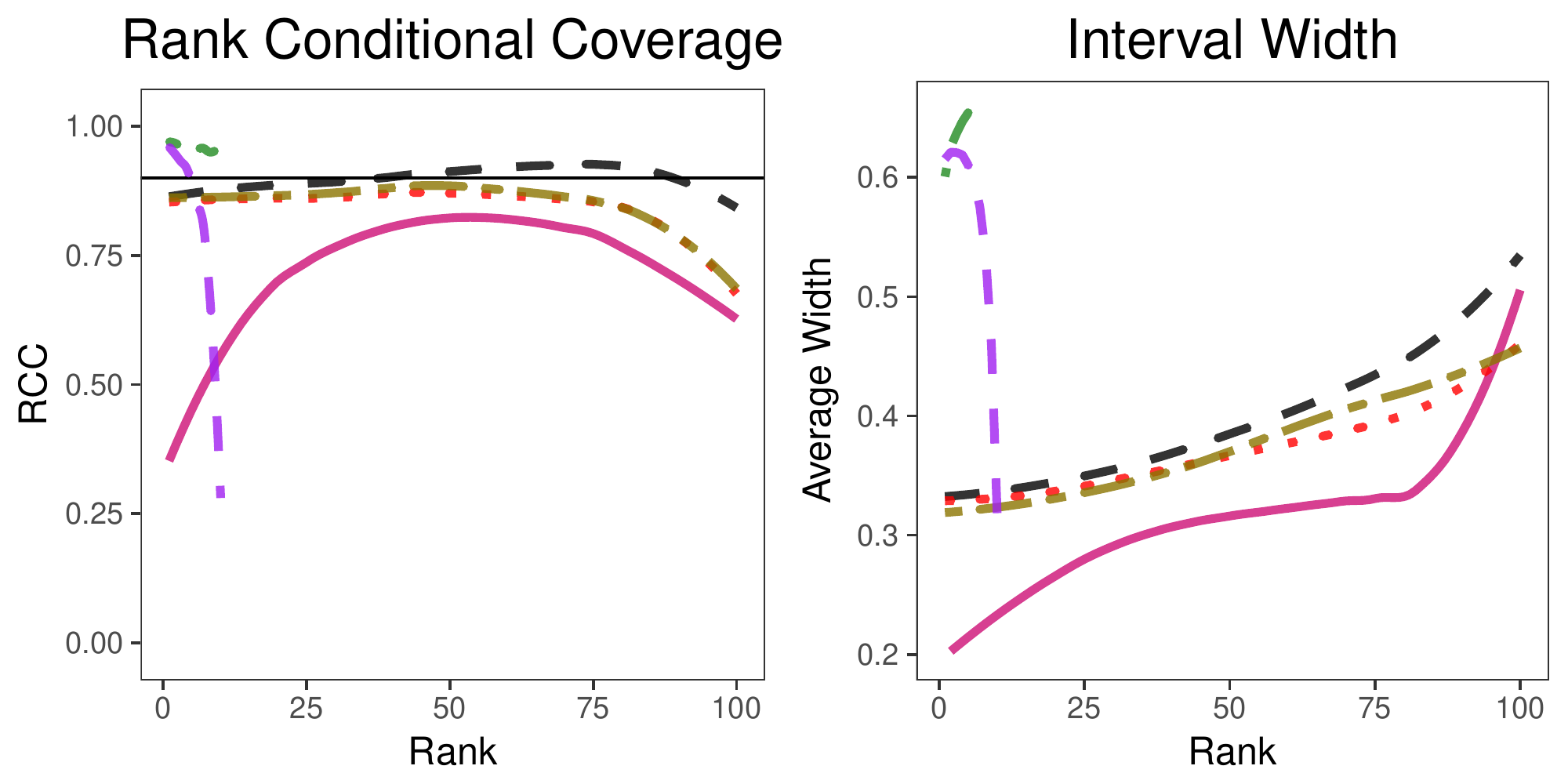}
\includegraphics[width=\textwidth]{blocks_legend.pdf}
\caption{Simulation results for Section~\ref{sec:sims_biomarker}. Rank conditional coverage (left) and interval widths (right) are shown for the top 20\% of parameters averaged over 400 simulations. Parameters are ranked by the absolute value of the test statistic. Coverage rates and widths are smoothed using loess. 
In the left-hand panel, a horizontal line shows the nominal level 90\%.
}
\label{fig:biomarker}
\end{figure}

\section{Discussion}
\label{sec:discussion}
Interval estimation when the number of parameters is large is a challenging problem often ignored in large scale studies. Out of caution, these studies are often limited to hypothesis testing but this limitation is unnecessary in many cases. We have shown that the full set of parameter estimates contains information and can be used to correct bias and generate useful confidence intervals. We have also introduced a more granular, informative concept of coverage which can be applied to confidence intervals constructed for numerous parameters. 

Rank conditional coverage is an important criterion to consider in evaluating confidence intervals for large parameter sets. As a finer grained criterion, it reveals problems that are ignored by the FCR. In many cases, using an FCR controlling procedure (even after selecting top parameter-estimates) results in very low coverage probabilities for the very largest parameters. %This is problematic when estimates and confidence intervals influence future research decisions such as sample size calculations for replication studies

In our simulations we found that rank conditional coverage is a more difficult criterion to control than the false coverage statement rate of \cite{Benjamini2005}. The two proposed bootstrapping methods almost always outperformed other methods and produced smaller intervals than all methods except the ashr method of \cite{Stephens2016}.

\bigskip
\begin{center}
{\large\bf SUPPLEMENTARY MATERIAL}
\end{center}

\begin{description}

\item[Appendix:] Contains a proof that controlling the RCC guarantees control of the FCR, two variations of Algorithm~\ref{alg:parametric}, and additional simulation results referenced in the text.

\item[R-package \texttt{rcc}:] R-package implementing Algorithm~\ref{alg:parametric},  Appendix Algorithms~\ref{alg:parametric2}, and \ref{alg:absparametric}, and Algorithm~\ref{alg:nonparametric}. (GNU zipped tar file)

\item[R-package \texttt{rcc-sims}:] R-package replicating the simulations shown in Section~\ref{sec:example}, and \ref{sec:sim}. (GNU zipped tar file)

\end{description}

\bibliographystyle{chicago}

\bibliography{ci}

\clearpage
\newpage

\section{Appendix}
\subsection{Proof that controlling RCC guarantees FCR control}
\label{sec:supp_thm}
\begin{thm}\label{thm:rcc_fcr}
If a set of confidence intervals $CI_{1}\dots CI_{p}$ controls RCC for a particular ranking scheme at level $\alpha$, then these intervals control FCR for any selection procedure which selects the top $r$ using the same ranking scheme.
\end{thm}
\begin{proof}
Let $S(\mathbf{t})$ be a selection procedure which chooses the $r$ most significant parameters based on the vector of statistics $\mathbf{t} = (t_{1} \dots t_{p})$. The number of selected parameters may be data dependent.
Using the identity from \cite{Benjamini2005}
\begin{align*}
FCR =& \sum_{r=1}^{p} \frac{1}{r}\sum_{i=1}^{p} P[\theta_{i} \not\in CI_{i}, i\in S, \vert S \vert = r]\\
=& \sum_{r=1}^{p} \frac{1}{r}\sum_{i=p-r-1}^{p} P[\theta_{s(i)} \not\in CI_{s(i)}, \vert S \vert = r]\\
= & \sum_{r=1}^{p} \frac{1}{r}\sum_{i=p-r-1}^{p} \alpha P[ \vert S \vert = r]\\
= & \alpha P[\vert S \vert \geq 1] \leq \alpha
\end{align*}
\end{proof}
\clearpage
\newpage

\subsection{Variations of the Bootstrapping Procedure in Algorithm~\ref{alg:parametric} of the Main Text}

Here we present two variations of the Parametric bootstrap presented in Algorithm~\ref{alg:parametric} of the main text. In Algorithm~\ref{alg:parametric2}, we show the additional steps necessary for bootstrapping from de-biased parameter estimates. In Algorithm~\ref{alg:parametric2} we show a variation of the algorithm for absolute value based ranking schemes. We define an absolute value based ranking scheme as any ranking scheme that is invariant to changes in the sign of any or all parameter estimates. For example, ranking estimates based on the absolute value of the $t$-statistic $\hat{\theta}_i/\hat{se}(\hat{\theta})_i$ is an absolute value based ranking scheme.
\begin{algorithm}[h!]
\caption{Simple parametric bootstrap with de-biased mean estimates}\label{alg:parametric2}
\begin{enumerate}[1.]
\item Obtain debiased estimates $\tilde{\theta}_{1} \dots \tilde{\theta}_{p}$ for $\theta_{1} \dots \theta_{p}$ for example using the method of \cite{Simon2013}.
\item For $k$ in $1 \dots K$:
\begin{enumerate}[a.]
\item Sample $\vartheta_{k,i}$ from a $N(\tilde{\theta}_{i}, 1)$ distribution for $i$ in $1\dots p$.
%\item Calculate debiased estimates from the bootstrapped sample $\tilde{\vartheta}_{k, 1} \dots \tilde{\vartheta}_{k,p}$
\item Estimate the bias at each rank $\hat{\delta}_{k,[i]}$ as 
\[
\hat{\delta}_{k,[i]} = \vartheta_{k,(i)} - \tilde{\theta}_{s_k(i)}
\]
\end{enumerate}
\item For $i$ in $1\dots p$:
\begin{enumerate}[a.]
\item Calculate empirical quantiles $\hat{H}^{-1}_{[i]}(x)$ of $\lbrace\hat\delta_{1,[i]}\dots \hat{\delta}_{K,[i]}\rbrace$
\item Generate $CI_{s(i)}^{\text{boot}}$ as in (\ref{eq:boot_ci})
\end{enumerate}
\end{enumerate}
\end{algorithm}

\begin{algorithm}[h!]
\caption{Simple parametric bootstrap based on absolute value ranking}\label{alg:absparametric}

\begin{enumerate}[1.]
\item For $k$ in $1 \dots K$:
\begin{enumerate}[a.]
\item Sample $\vartheta_{k,i}$ from a $N(\hat{\theta}_{i}, 1)$ for $i$ in $1\dots p$.  Define $\vartheta_{(i)}$ such that 
\[
\vert \vartheta_{(1)}\vert \geq \dots \geq \vert\vartheta_{(p)}\vert
\]
\item Calculate the bias at each rank $\hat{\delta}_{k,[i]}$ as 
\[
\hat{\delta}_{k,[i]} = \text{sign}(\vartheta_{k,(i)})(\vartheta_{k,(i)} - \theta_{s_k(i)})
\]
\end{enumerate}
\item For $i$ in $1\dots p$:
\begin{enumerate}[a.]
\item Calculate empirical quantiles $\hat{H}^{-1}_{[i]}(x)$ of $\lbrace \hat{\delta}_{1,[i]}\dots \delta_{K,[i]}\rbrace$
\item Generate $CI_{s(i)}^{\text{boot}}$ as 
\begin{align*}
\begin{cases}
\Big(\hat{\theta}_{(i)}-\hat{H}^{-1}_{[i]}(1-\alpha/2),\ \hat{\theta}_{(i)} -\hat{H}^{-1}_{[i]}(\alpha/2)\Big) \qquad & \hat{\theta}_{s(i)} > 0\\
\Big(\hat{\theta}_{(i)}+\hat{H}^{-1}_{[i]}(\alpha/2),\ \hat{\theta}_{(i)} + \hat{H}^{-1}_{[i]}(1-\alpha/2)\Big) \qquad & \hat{\theta}_{s(i)} \leq 0
\end{cases}
\end{align*}
\end{enumerate}
\end{enumerate}
\end{algorithm}

\clearpage
\subsection{Simulations Results Using a Block-Based Ranking Scheme}\label{sec:sims_block}

Here we consider the Simulations described in Section~\ref{sec:sims_lin} of the main text, but instead of simply ranking by test statistic, we use a ranking scheme that incorporates the block correlation structure of the features. 
In Section~\ref{sec:sims_lin} of the main text, we describe how 1000 features are simulated in 100 blocks each containing 10 correlated features. Consider a ranking scheme in which we first select the most significant feature in each block and then rank these selected features. 

This is similar to the way that genome-wide association study results are often presented. In these studies, we often discover that many variants in a small genomic region are associated with the trait. Usually this arises because there are many variants that are correlated with a single causal variant so researches will typically describe only the top variant in the region.

Rank conditional coverage and interval widths averaged over 400 simulations for each scenario are shown in Figure~\ref{fig:linreg_cw}. 
The methods of \cite{Reid2014} and \cite{Weinstein2013} are not shown because these are based on selection rules based on the test statistic alone. These results show a similar pattern as those presented in Section~\ref{sec:sims_lin} of the main text --- the parametric and non-parametric boostrap confidence intervals have close to the nominal RCC and are shorter than the naive intervals. Both ash and the naive intervals have much lower RCC than the nominal level. In the scenario with the largest correlation between features ($\rho=0.8$), we find undercoverage of the least significant parameters by the parametric bootstrap. This is caused by correlation between the parameter estimates which are assumed to be independent in this implementation of the parametric bootstrap.

\begin{figure}
\centering
Average Rank Conditional Coverage\\
\includegraphics[width=\textwidth]{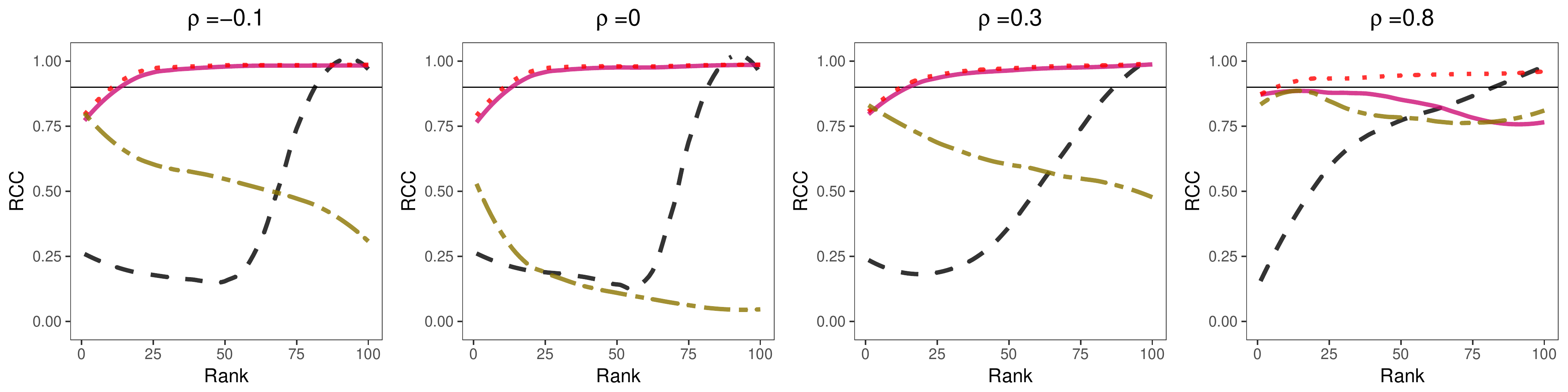}\\
Interval Width\\
\includegraphics[width=\textwidth]{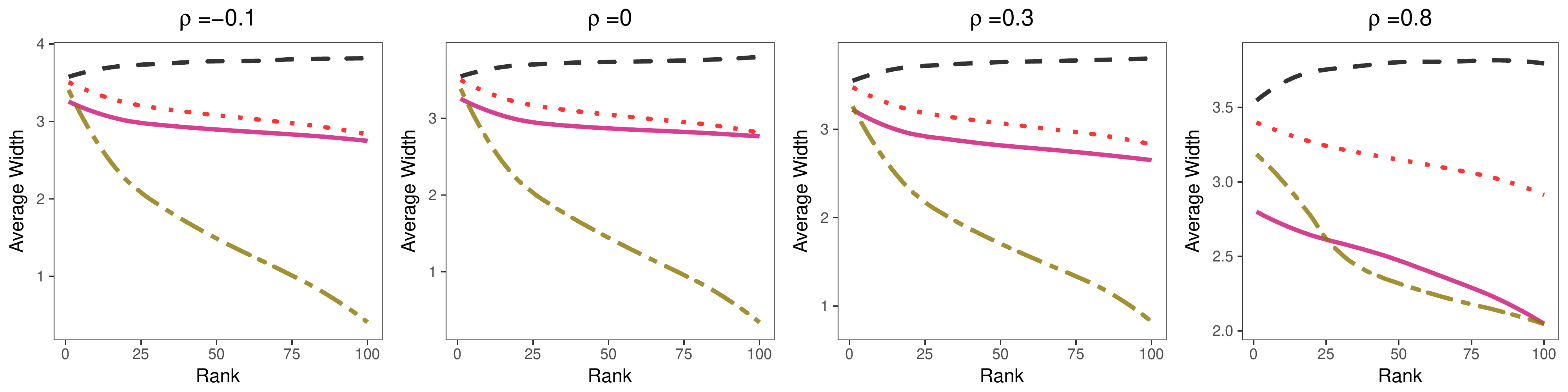}\\
\includegraphics[width=0.667\textwidth]{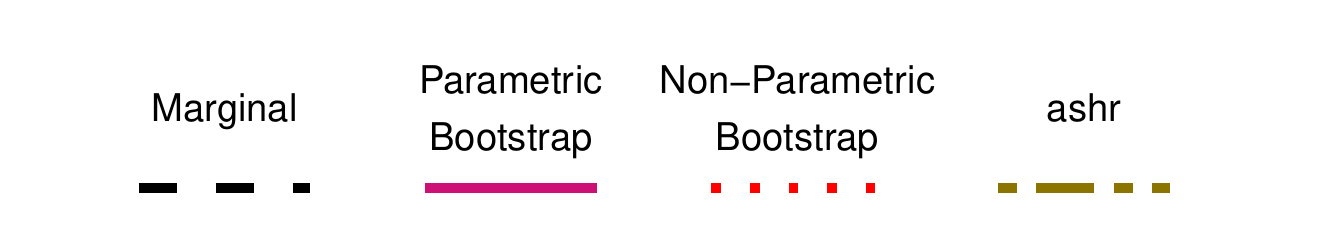}

\caption{Simulation results for Appendix Section~\ref{sec:sims_block}. Rank conditional coverage (top) and interval widths (bottom) are shown for all 100 ranked parameters averaged over 400 simulations. Parameters are ranked by first choosing the most significant parameter in each block and then ranking the selected parameters based on the absolute value of the test statistic. Coverage rates and widths are smoothed using loess. 
In the top panel, a horizontal line shows the nominal level 90\%.
}
\label{fig:linreg_cw}
\end{figure}

\end{document}